\documentclass[onecolumn,amsmath,amssymb,10pt,aps]{revtex4}
  
\usepackage[utf8]{inputenc}
\usepackage[T1]{fontenc}

\usepackage{amsmath}
\usepackage{amsfonts}
\usepackage{graphicx}
\usepackage{yfonts}
\usepackage{color}
\usepackage[normalem]{ulem}
\usepackage{amsthm}
\usepackage{bm}
\usepackage{bbm}
\usepackage{mathtools}
\usepackage{array}
\usepackage{placeins}
\usepackage{enumitem}
\usepackage{tikz}

\def\<{\langle}
\def\>{\rangle}

\newcommand{\Tr}{\mathrm{Tr}}
\def\oper{{\mathchoice{\rm 1\mskip-4mu l}{\rm 1\mskip-4mu l}
{\rm 1\mskip-4.5mu l}{\rm 1\mskip-5mu l}}}
\DeclareMathAlphabet\mathbfcal{OMS}{cmsy}{b}{n}

\mathchardef\mhyphen="2D % Define a "math hyphen"

\newtheorem{Definition}{Definition}

\newtheorem{Proposition}{Proposition}
\newtheorem{Example}{Example}

\begin{document}

\title{Indecomposable entanglement witnesses from symmetric measurements}

\author{Katarzyna Siudzi\'{n}ska}
\affiliation{Institute of Physics, Faculty of Physics, Astronomy and Informatics \\  Nicolaus Copernicus University in Toru\'{n}, ul. Grudzi\k{a}dzka 5, 87--100 Toru\'{n}, Poland}

\begin{abstract}
We propose a family of positive maps constructed from a recently introduced class of symmetric measurements. These maps are used to define entanglement witnesses, which include other popular approaches with mutually unbiased bases and mutually unbiased measurements. A particular interest is given to indecomposable witnesses that can be used to detect entanglement of quantum states with positive partial transposition. We present several examples for different number of measurements.
\end{abstract}

\flushbottom

\maketitle

\thispagestyle{empty}

\section{Introduction}

Quantum entanglement is an essential resource that plays a fundamental role in quantum information processing, quantum communication, quantum computing, and other modern quantum technologies \cite{HHHH,Nielsen}. In particular, any bipartite entangled state enhances the teleportation power \cite{Masanes} and displays hidden nonlocality \cite{Masanes2}. The usefulness of quantum tasks usually increases with the amount of entanglement \cite{Takagi,SkrzypczykLinden,BaeDarek}. Characterization of entangled states is essential both in theory and practice. However, the problem of distinguishing between separable and entanged states remains open; in fact, it is NP-hard \cite{Gurvits}.

For qubit-qubit and qubit-qutrit systems, the necessary and sufficient separability condition is given by the celebrated Peres-Horodecki positive partial transposition (PPT) criterion \cite{PhysRevLett.77.1413,HORODECKI1997333}. In higher dimensions, this condition is only necessary, which first was shown for a qutrit-qutrit system \cite{HORODECKI1997333}. More refined detection methods include the computable cross-norm or realignment (CCNR) criterion \cite{Rudolph,HHH_OSID,ChenWu,ChenWu2}, the correlation matrix criterion \cite{deVicente,deVicente2}, the local uncertainty relationship criterion \cite{Takeuchi}, the reduced density matrix criterion \cite{Gingrich}, and the covariance matrix criterion \cite{CMM1}.

Another approach to entanglement detection is through entanglement witnesses, which are Hermitian block-positive (but not positive) operators. Hence, any such operator is positive on separable states, and a state $\rho$ is separable if and only if $\Tr(\rho W)\geq 0$ for every entanglement witness $W$. All entangled states have witnesses that detect them \cite{Terhal1,Terhal2}. In other words, if $\rho$ is entangled, there exist a (non-unique) witness $W$ of its entanglement such that $\Tr(\rho W)<0$. The problem lies in finding a suitable witness for a given state. The advantage of choosing entanglement witnesses over other detection methods is that non-separability of a state is decided upon calculating expectation values of $W$ in that state. Therefore, it requires less information than a full state tomography, which also means less experimental devices and fewer measurements performed.

There exists a special class of witnesses that can detect quantum states with positive partial transposition, also known as bound entangled states \cite{Horodeccy,Horodeccy2,Terhal2,Lewenstein,Lewenstein2}. They are called indecomposable because they cannot be decomposed into $W=A+B^\Gamma$ with positive $A$ and $B$, where $\Gamma$ is a partial transposition. There is no general construction method for such operators, and it is often hard to determine whether a witness is decomposable or not. However, several classes of indecomposable entanglement witnesses have been found, like the ones related to the well-known realignment or computable cross-norm (CCNR) separability criterion \cite{R1,R2,R3} and covariance matrix criterion \cite{CMM1,CMM2,CMM3}, as well as their generalizations \cite{ESIC-Darek,ESIC-Darek2}.

In the construction of entanglement witnesses, one often uses mutually unbiased bases (MUBs). Orthonormal bases in $\mathbb{C}^d$ are mutually unbiased if and only if the transition probability between any two vectors that belong to different bases is constant \cite{Durt}. In ref. \cite{MUBs}, the authors used MUBs to define a new class of witnesses and analyzed their properties in $d=3$. This construction has been generalized in many ways. Li et al. introduced analogical operators for mutually unbiased measurements (MUMs) \cite{Li} and symmetric informationally complete measurements (SIC-POVMs) \cite{EW-SIC}. Wang and Zheng \cite{EW-2MUB} considered the MUB-based witnesses in composite systems with different dimensions. Hiesmayr et al. \cite{Junu2} showed that inequivalent and unextendible sets of MUBs are sometimes more useful to detect entanglement, whereas Bae et al. \cite{HowMany} found that more than $d/2+1$ MUBs are needed to identify bound entangled states. The MUMs encompassing a full range of purity were shown to detect entanglement with as little as two measurement operators \cite{MUM_purity}. Recently, it has also been proven that the witnesses constructed from MUMs belong to the class based on the CCNR criterion \cite{P_maps}.

In this paper, we construct a class of entanglement witnesses based on symmetric measurements introduced in ref. \cite{SIC-MUB}. These measurements include mutually unbiased bases and symmetric informationally complete POVMs as special cases. In any finite dimension $d>2$, one has at least four informationally complete sets of symmetric measurements. Their ability to detect quantum entanglement has already been shown \cite{SIC-MUB}. We recall the definition and basic properties of such operators, together with the construction method. Next, we use a (not necessarily complete) set of measurements to establish a family of positive, trace-preserving maps and the associated entanglement witness. These witnesses are also related to the class based on the CCNR criterion \cite{YuLiu}. Also, they do not depend on the purity of measurement operators but on their number inside each POVM. It turns out that one can construct the same entanglement witness using different measurements. Finally, we provide examples of indecomposable witnesses that belong to our class.

\section{Classes of symmetric measurements}

Quantum measurements are represented by positive, operator-valued measures (POVMs), which are families of positive operators that sum up to the identity.  For a given density operator $\rho$, the probability outcome associated with a POVM element $E_\alpha$ is $p_\alpha=\Tr(E_\alpha\rho)$. Recently, a general class of symmetric measurements has been introduced \cite{SIC-MUB}. Let us summarize its definition and main properties.

\begin{Definition}
An $(N,M)$-POVM is a collection of $N$ POVMs $\{E_{\alpha,k};\,k=1,\ldots,M\}$, $\alpha=1,\ldots,N$, that satisfy the symmetry conditions
\begin{equation}\label{M}
\begin{split}
\Tr (E_{\alpha,k})&=\frac dM,\\
\Tr (E_{\alpha,k}^2)&=x,\\
\Tr (E_{\alpha,k}E_{\alpha,\ell})&=\frac{d-Mx}{M(M-1)},\qquad \ell\neq k,\\
\Tr (E_{\alpha,k}E_{\beta,\ell})&=\frac{d}{M^2},\qquad \beta\neq\alpha,
\end{split}
\end{equation}
with a free parameter $x$ that belongs to the range
\begin{equation}\label{x}
\frac{d}{M^2}<x\leq\min\left\{\frac{d^2}{M^2},\frac{d}{M}\right\}.
\end{equation}
\end{Definition}

If $x=d^2/M^2$, then the $(N,M)$-POVM describes projective measurements. Moreover, eq. (\ref{x}) implies that there are no projective $(N,M)$-POVMs for $M<d$. Symmmetric measurements are informationally complete if and only if
\begin{equation}\label{mn}
N=\frac{d^2-1}{M-1}.
\end{equation}
Observe that if $d=2$, there are only two possible choices of admissible $M$ and $N$: $(N=1,M=4)$ for the general symmetric, informationally complete (SIC) POVMs \cite{Gour} and $(N=3,M=2)$ for mutually unbiased measurements (MUMs) \cite{Kalev}. In this case, Definition \ref{M} only provides a unified method to describe generalizations of SIC POVMs and mutually unbiased bases (MUBs). However, this changes for $d>2$.

\begin{Proposition}\label{P1}
In any dimension $d<\infty$, there exist at least four distinct classes of informationally complete $(N,M)$-POVMs depending on the choice of $M$ and $N$:
\begin{enumerate}[label=(\roman*)]
\item $M=d^2$ and $N=1$ (general SIC POVM),
\item $M=d$ and $N=d+1$ (MUMs),
\item $M=2$ and $N=d^2-1$,
\item $M=d+2$ and $N=d-1$.
\end{enumerate}
\end{Proposition}

Informationally complete $(N,M)$-POVMs can be constructed in arbitrary dimensions using orthonormal Hermitian operator bases $\{G_0=\mathbb{I}_d/\sqrt{d},G_{\alpha,k};\,\alpha=1,\ldots,N,\,k=1,\ldots,M-1\}$ with $\Tr G_{\alpha,k}=0$. Namely,
\begin{equation}\label{E}
E_{\alpha,k}=\frac 1M \mathbb{I}_d+tH_{\alpha,k},
\end{equation}
where
\begin{equation}\label{H}
H_{\alpha,k}=\left\{\begin{aligned}
&G_\alpha-\sqrt{M}(\sqrt{M}+1)G_{\alpha,k},\quad k=1,\ldots,M-1,\\
&(\sqrt{M}+1)G_\alpha,\qquad k=M,
\end{aligned}\right.
\end{equation}
and $G_\alpha=\sum_{k=1}^{M-1}G_{\alpha,k}$. The parameter $t$ is related to $x$ via
\begin{equation}\label{xt}
x=\frac{d}{M^2}+t^2(M-1)(\sqrt{M}+1)^2.
\end{equation}
The optimal value $x_{\rm opt}$, which is the greatest $x$ such that $E_{\alpha,k}\geq 0$, depends on the operator basis.

Symmetric measurements admit an important property that becomes useful later on.

\begin{Proposition}
For $L\leq N$ and an arbitrary state $\rho$, the elements of an $(N,M)$-POVM satisfy
\begin{equation}\label{C}
\sum_{\alpha=1}^L\sum_{k=1}^Mp_{\alpha,k}^2\leq\frac LM +\frac{(M^2x-d)[d\Tr(\rho^2)-1]}{dM(M-1)},
\end{equation}
where $p_{\alpha,k}=\Tr(E_{\alpha,k}\rho)$.
\end{Proposition}

\begin{proof}
We follow the method presented in ref. \cite{Rastegin2}. First, observe that $H_{\alpha,k}$ are traceless operators that satisfy the relations
\begin{equation}
\begin{split}
\Tr(H_{\alpha,k}^2)&=(M-1)(\sqrt{M}+1)^2,\\
\Tr(H_{\alpha,k}H_{\alpha,\ell})&=-(\sqrt{M}+1)^2,\qquad \ell\neq k,\\
\Tr(H_{\alpha,k}H_{\beta,\ell})&=0,\qquad \beta\neq\alpha.
\end{split}
\end{equation}
Now, similarly to $G_{\alpha,k}$, $H_{\alpha,k}$ span the space of traceless Hermitian operators. Therefore, any quantum state $\rho$ can be represented as
\begin{equation}
\rho=\frac 1d \mathbb{I}_d+\sum_{\alpha=1}^N\sum_{k=1}^Mr_{\alpha,k}H_{\alpha,k}
\end{equation}
with real-valued parameters $r_{\alpha,k}$. Next, we calculate
\begin{equation}
p_{\alpha,k}=\Tr(\rho E_{\alpha,k})=\frac 1M + t\Tr(\rho H_{\alpha,k})=
\frac 1M + t\sum_{\ell=1}^Mr_{\alpha,k}\Tr(H_{\alpha,k}H_{\alpha,\ell})
=\frac 1M +t(\sqrt{M}+1)^2(Mr_{\alpha,k}-r_\alpha),
\end{equation}
where $r_\alpha=\sum_{k=1}^Mr_{\alpha,k}$. The sum of squares reads
\begin{equation}
\sum_{k=1}^Mp_{\alpha,k}^2=\frac 1M+t^2M(\sqrt{M}+1)^4\left(M\sum_{k=1}^M r_{\alpha,k}^2-r_\alpha^2\right).
\end{equation}
Finally, we take the sum over $\alpha=1,\ldots,L$ and get
\begin{equation}\label{almost}
\sum_{\alpha=1}^L\sum_{k=1}^Mp_{\alpha,k}^2=\frac LM+t^2M(\sqrt{M}+1)^4\sum_{\alpha=1}^L\left(M\sum_{k=1}^M r_{\alpha,k}^2-r_\alpha^2\right).
\end{equation}
It remains to notice that
\begin{equation}
\Tr(\rho^2)-\frac 1d
=\sum_{\alpha=1}^N\sum_{k,\ell=1}^Mr_{\alpha,k}r_{\alpha,\ell}\Tr(H_{\alpha,k}H_{\alpha,\ell})
\geq
\sum_{\alpha=1}^L\sum_{k,\ell=1}^Mr_{\alpha,k}r_{\alpha,\ell}\Tr(H_{\alpha,k}H_{\alpha,\ell})
=(\sqrt{M}+1)^2\left[M\sum_{k=1}^Mr_{\alpha,k}^2-r_\alpha^2\right].
\end{equation}
Applying this result in eq. (\ref{almost}), we arrive at
\begin{equation}
\sum_{\alpha=1}^L\sum_{k=1}^Mp_{\alpha,k}^2\leq\frac LM +t^2M(\sqrt{M}+1)^2\left[\Tr(\rho^2)-\frac 1d \right]
=\frac LM +\frac{(M^2x-d)[d\Tr(\rho^2)-1]}{dM(M-1)}
\end{equation}
after using the correspondence between $t$ and $x$ in eq. (\ref{xt}).
\end{proof}

For $L=N$, the relation in eq. (\ref{C}) is an equality and it reproduces the result first obtained in ref. \cite{SIC-MUB},
\begin{equation}
\sum_{\alpha=1}^N\sum_{k=1}^Mp_{\alpha,k}^2=\frac{d(M^2x-d)\Tr(\rho^2)+d^3-M^2x}{dM(M-1)}.
\end{equation}
On the other hand, for pure states $\rho$, eq. (\ref{C}) reduces to
\begin{equation}\label{ineq}
\sum_{\alpha=1}^L\sum_{k=1}^Mp_{\alpha,k}^2\leq
\frac LM +\frac{(d-1)(M^2x-d)}{dM(M-1)}.
\end{equation}

\section{Positive maps and entanglement witnesses}

Using any $(N,M)$-POVM (that is not necessarily informationally complete), define $N$ trace-preserving maps
\begin{equation}
\Phi_\alpha[X]=\frac Md \sum_{k,\ell=1}^{M}\mathcal{O}^{(\alpha)}_{k\ell}E_{\alpha,k}\Tr(XE_{\alpha,\ell}),
\end{equation}
where $\mathcal{O}^{(\alpha)}$ are orthogonal rotations that preserve the maximally mixed vector $\mathbf{n}_\ast=(1,\ldots,1)/\sqrt{d}$.

\begin{Proposition}\label{Prop}
The linear map
\begin{equation}\label{PTP}
\Phi=\frac{1}{b}\left[a\Phi_0+\sum_{\alpha=L+1}^N\Phi_\alpha-\sum_{\alpha=1}^L\Phi_\alpha\right]
\end{equation}
with $a=b-N+2L$ and $b=(d-1)M(x-y)/d$ is positive and trace-preserving. $\Phi_0$ is the maximally depolarizing channel.
\end{Proposition}

\begin{proof}
It suffices to prove that
\begin{equation}\label{pos}
\Tr(\Phi[P])^2\leq\frac{1}{d-1}
\end{equation}
holds for every rank-1 projectors $P$ \cite{MUBs}. For simplicity, introduce the map $\widetilde{\Phi}=b\Phi$, and then calculate
\begin{equation}\label{PTP1}
\begin{split}
\Tr(\widetilde{\Phi}[P])^2=\Tr\Bigg\{&a^2\Phi_0[P]^2+\sum_{\alpha,\beta=1}^L
\Phi_\alpha[P]\Phi_\beta[P]+\sum_{\alpha,\beta=L+1}^N\Phi_\alpha[P]\Phi_\beta[P]
+2a\sum_{\alpha=L+1}^N\Phi_0[P]\Phi_\alpha[P]\\&-2a\sum_{\alpha=1}^L\Phi_0[P]\Phi_\alpha[P]
-2\sum_{\alpha=L+1}^N\sum_{\beta=1}^L\Phi_\alpha[P]\Phi_\beta[P]\Bigg\}.
\end{split}
\end{equation}
First, let us simplify the terms in the brackets using the properties of $(N,M)$-POVMs and the orthogonal rotation matrices,
\begin{equation}
\sum_{k=1}^M\mathcal{O}^{(\alpha)}_{k\ell}=\sum_{\ell=1}^M\mathcal{O}^{(\alpha)}_{k\ell}=1,\qquad
\sum_{k=1}^M\mathcal{O}^{(\alpha)}_{k\ell}\mathcal{O}^{(\alpha)}_{km}=\delta_{\ell m}.
\end{equation}
One has
\begin{equation}
\Tr(\Phi_0[P]^2)=\Tr(\Phi_0[P]\Phi_\alpha[P])=\Tr(\Phi_\alpha[P]\Phi_\beta[P])=\frac 1d,\quad \alpha\neq\beta,
\end{equation}
as well as
\begin{equation}
\Tr(\Phi_\alpha[P]^2)=\frac{M}{d^2(M-1)}\left\{d-Mx+(M^2x-d)
\sum_{k=1}^M\left[\Tr(PE_{\alpha,k})\right]^2\right\}.
\end{equation}
Now, eq. (\ref{PTP1}) can be rewritten as
\begin{equation}\label{PTP2}
\begin{split}
\Tr(\widetilde{\Phi}[P])^2=&\,\frac 1d \Big[a^2+(N-L)(N-L-1)+L(L-1)+2a(N-L)-2aL
-2L(N-L)\Big]\\&+\frac{M}{d^2(M-1)}\left\{N(d-Mx)+(M^2x-d)
\sum_{\alpha=1}^N\sum_{k=1}^M\left[\Tr(PE_{\alpha,k})\right]^2\right\}\\
=&\frac 1d \Big[(a+N-2L)^2-N\Big]+\frac{M}{d^2(M-1)}\left\{N(d-Mx)+(M^2x-d)
\sum_{\alpha=1}^N\sum_{k=1}^M\left[\Tr(PE_{\alpha,k})\right]^2\right\}.
\end{split}
\end{equation}
Finally, eq. (\ref{C}) from Proposition 2 and the definition of $a$ allow us to write
\begin{equation}
\begin{split}
\Tr(\widetilde{\Phi}[P])^2\leq\frac{b^2-N}{d} +\frac{MN(d-Mx)}{d^2(M-1)}+\frac{M(M^2x-d)}{d^2(M-1)}\left[
\frac NM +\frac{(d-1)(M^2x-d)}{dM(M-1)}\right]=
\frac{b^2}{d-1},
\end{split}
\end{equation}
which finally proves the validity of condition (\ref{pos}).
\end{proof}

Positive maps are used in the theory of quantum entanglement to detect entangled states; that is, states that are not separable. Positive but not completely positive maps $\Phi$ define entanglement witnesses $W$ through the Choi-Jamio{\l}kowski isomorphism,
\begin{equation}\label{W}
W=\sum_{k,\ell=0}^{d-1}|k\>\<\ell|\otimes\Phi[|k\>\<\ell|],
\end{equation}
where $|k\>$ is an orthonormal basis in $\mathbb{C}^d$. In particular, the positive map $\Phi$ from Proposition 3 gives rise to
\begin{equation}\label{W2}
W=\frac{1}{b}\left(\frac ad \mathbb{I}_{d^2}
+\sum_{\alpha=L+1}^NK_\alpha-\sum_{\alpha=1}^LK_\alpha\right),
\end{equation}
where
\begin{equation}
K_\alpha=\frac Md \sum_{k,\ell=1}^M\mathcal{O}_{k\ell}^{(\alpha)}
\overline{E}_{\alpha,\ell}\otimes E_{\alpha,k}.
\end{equation}
In any dimension $d$, one can always construct informationally complete $(N,M)$-POVMs by using an orthonormal basis $\{\mathbb{I}_d/\sqrt{d},G_{\alpha,k}\}$ of traceless Hermitian operators $G_{\alpha,k}$. Then, the entanglement witness can be expressed in terms of the operator basis,
\begin{equation}\label{WW}
\widetilde{W}=\frac{b}{t^2} W = \frac{d-1}{d^2}M^2(\sqrt{M}+1)^2
\mathbb{I}_{d^2}+\sum_{\alpha=L+1}^NJ_\alpha-\sum_{\alpha=1}^LJ_\alpha,
\end{equation}
where
\begin{equation}\label{J}
J_\alpha=\frac Md \sum_{k,\ell=1}^{M}\mathcal{O}_{kl}^{(\alpha)}
\overline{H}_{\alpha,\ell}\otimes H_{\alpha,k}.
\end{equation}
Note that this witness does not depend on the parameter $x$ that characterizes symmetric measurements. However, $\widetilde{W}$ does depend on the number $M$ of operators inside a single POVM. The greater the value of $M$, the greater $L$ can be (the more can be subtracted).

Going a step further, $\widetilde{W}$ can be directly written in terms of the operator basis elements $G_{\alpha,k}$. Using eqs. (\ref{H}) and (\ref{J}), it is straightforward to find
\begin{equation}
J_\alpha=\frac Md \sum_{k,\ell=1}^{M-1}\mathcal{Q}_{kl}^{(\alpha)}\overline{G}_{\alpha,\ell}
\otimes G_{\alpha,k}
\end{equation}
with
\begin{equation}\label{Q}
\mathcal{Q}_{k\ell}^{(\alpha)} = M(\mathcal{O}_{MM}^{(\alpha)}-1) + M(\sqrt{M}+1)^2 \mathcal{O}_{k\ell}^{(\alpha)} - M(\sqrt{M}+1) (\mathcal{O}_{M\ell}^{(\alpha)}+\mathcal{O}_{kM}^{(\alpha)}).
\end{equation}
Such $\mathcal{Q}^{(\alpha)}$ are rescaled orthogonal matrices since they satisfy
\begin{equation}
\mathcal{Q}^{(\alpha)T}\mathcal{Q}^{(\alpha)}
=\mathcal{Q}^{(\alpha)}\mathcal{Q}^{(\alpha)T}=M^2(\sqrt{M}+1)^4\mathbb{I}_{M-1}.
\end{equation}
For the special case with $\mathcal{O}^{(\alpha)}=\mathbb{I}_M$, one has
\begin{equation}\label{Id}
\mathcal{Q}_{k\ell}^{(\alpha)} = M(\sqrt{M}+1)^2 \delta_{k\ell}.
\end{equation}

Now, let us show how $\widetilde{W}$ from eq. (\ref{WW}) relates to a well-known class of entanglement witnesses. For now, assume that the $(N,M)$-POVM is informationally complete and $L=N$. In this case, the associated witness reads
\begin{equation}
\widetilde{W}=\frac{M^2}{d}(\sqrt{M}+1)^2\left[\mathbb{I}_{d^2}- G_0\otimes G_0-\frac{d}{M^2(\sqrt{M}+1)^2}\sum_{\alpha=1}^{N}J_\alpha\right],
\end{equation}
where $G_0=\mathbb{I}_d/\sqrt{d}$. A simple relabelling of indices $(\alpha,k)\longmapsto\mu$ allows us to write
\begin{equation}
\widetilde{W}^\prime=\frac{d\widetilde{W}}{M^2(\sqrt{M}+1)^2}=\mathbb{I}_{d^2}
-\sum_{\mu,\nu=0}^{d^2-1}Q_{\mu\nu} G_\mu^T\otimes G_\nu,
\end{equation}
where we introduced the block-diagonal orthogonal matrix
\begin{equation}
Q=\frac{1}{M(\sqrt{M}+1)^2}
\left[\begin{array}{c c c c c}
M(\sqrt{M}+1)^2 & & & & \\
 & \mathcal{Q}^{(1)T} & & & \\
 & & \mathcal{Q}^{(2)T} & & \\
 & & & \ddots & \\
 & & & & \mathcal{Q}^{(N)T}
\end{array}\right].
\end{equation}
Now, it becomes evident that entanglement witnesses $\widetilde{W}$ constructed from symmetric measurements are a part of a larger class
\begin{equation}\label{WW2}
W^\prime=\mathbb{I}_{d^2}-\sum_{\mu,\nu=0}^{d^2-1}Q_{\mu\nu} G_\mu^T\otimes G_\nu,
\end{equation}
which is related to the CCNR criterion \cite{YuLiu}. In the above formula, $G_\mu$ are the elements of an arbitrary orthonormal Hermitian basis and $Q_{\mu\nu}$ is any orthogonal matrix ($Q^TQ=\mathbb{I}_{d^2}$). However, it is enough that $Q^TQ\leq\mathbb{I}_{d^2}$. Hence, our discussion also holds true for any $(N,M)$-POVMs and $L\leq N$, provided that the rotation matrices $\mathcal{O}^{(\alpha)}$ get replaced with $\widetilde{\mathcal{O}}^{(\alpha)}$ that are allowed to change the sign of $\mathbf{n}_\ast$ ($\widetilde{\mathcal{O}}^{(\alpha)}\mathbf{n}_\ast=\pm\mathbf{n}_\ast$) or vanish ($\widetilde{\mathcal{O}}^{(\alpha)}=\mathbb{O}_M$).

\begin{Example}
Take any informationally complete $(N,M)$-POVM. From eq. (\ref{Id}), it follows that fixing all $\mathcal{O}^{(\alpha)}=\mathbb{I}_M$ and $L=N$ always results in $Q=\mathbb{I}_{d^2}$. The corresponding entanglement witness has the form
\begin{equation}
\widetilde{W}^\prime
=\mathbb{I}_{d^2}-\sum_{\mu=0}^{d^2-1} G_\mu^T\otimes G_\mu.
\end{equation}
Therefore, it is possible to produce the same witnesses $\widetilde{W}^\prime$ using different $(N,M)$-POVMs, provided that they arise from the same Hermitian orthonormal basis. In particular, for the Gell-Mann matrices, one recovers the reduction map $\widetilde{W}^\prime=\mathbb{I}_{d^2}-dP_+$ \cite{EW_reduction,P_maps}.
\end{Example}

\begin{Example}\label{Ex2}
If $M=2$, then the only admissible choices for the rotation matrices are $\mathcal{O}^{(\alpha)}=\mathbb{I}_2$ or $\mathcal{O}^{(\alpha)}=\sigma_1$. This means that all witnesses constructed from $(N,2)$-POVMs have the form
\begin{equation}\label{WEx2}
\widetilde{W}^\prime=
\mathbb{I}_{d^2}+\sum_{\alpha=L+1}^N G_\alpha^T\otimes G_\alpha-\sum_{\alpha=0}^L G_\alpha^T\otimes G_\alpha,
\end{equation}
where $N\leq d^2-1$.
\end{Example}

%\section{Case study: $d=3$}

\section{Indecomposabile witnesses in $d=3$}

Let us consider composite quantum systems with the underlying Hilbert space $\mathcal{H}\simeq\mathbb{C}^3\otimes\mathbb{C}^3$, which corresponds to fixing $d=3$. These are the lowest dimensional bipartite systems where the Peres-Horodecki PPT criterion is only necessary for separability. Therefore, more sophisticated methods of entanglement detection are needed. An important class of entanglement witnesses are indecomposable witnesses, which detect more entanglement than transposition $T$. Recall that a witness is decomposable if and only if it can be represented as $W=A+B^\Gamma$ with $A$, $B$ being positive operators and $\Gamma=\oper\otimes T$ denoting a partial transposition. Otherwise, it is indecomposable. Construction of indecomposable witnesses is hard but rewarding, as they can detect quantum states with positive partial transposition (PPT states). Notably, to prove indecomposability of $W$, it is enough to show a single PPT state that it detects. Below, we present several classes that can be obtained from symmetric measurements.

\begin{Example}\label{Ex1}
First, consider a full set of $N=4$ mutually unbiased measurements ($M=3$) constructed from the Gell-Mann matrices (see Appendix A). Note that for this choice of an orthonormal basis, the parameter $x=5/9<x_{\rm opt}=1$ is not optimal. Now, for $L=3$ and the permutation matrices
\begin{equation}
\mathcal{O}^{(\alpha)}=\begin{pmatrix}
0 & 1 & 0 \\
0 & 0 & 1 \\
1 & 0 & 0
\end{pmatrix},
\end{equation}
the associated entanglement witness
\begin{equation}
\widetilde{W}_1=\frac 16 \left[\begin{array}{c c c|c c c|c c c}
2 & \cdot & \cdot & \cdot & 3(1-i\sqrt{3}) & \cdot & \cdot & \cdot & 3(1-i\sqrt{3}) \\
\cdot & 2 & \cdot & \cdot & \cdot & \cdot & \cdot & \cdot & \cdot \\
\cdot & \cdot & 8 & \cdot & \cdot & \cdot & \cdot & \cdot & \cdot \\
\hline
\cdot & \cdot & \cdot & 8 & \cdot & \cdot & \cdot & \cdot & \cdot \\
3(1+i\sqrt{3}) & \cdot & \cdot & \cdot & 2 & \cdot & \cdot & \cdot & 3(1-i\sqrt{3}) \\
\cdot & \cdot & \cdot & \cdot & \cdot & 2 & \cdot & \cdot & \cdot \\
\hline
\cdot & \cdot & \cdot & \cdot & \cdot & \cdot & 2 & \cdot & \cdot \\
\cdot & \cdot & \cdot & \cdot & \cdot & \cdot & \cdot & 8 & \cdot \\
3(1+i\sqrt{3}) & \cdot & \cdot & \cdot & 3(1+i\sqrt{3}) & \cdot & \cdot & \cdot & 2
\end{array}\right]
\end{equation}
is indecomposable, as it detects a PPT state
\begin{equation}
\rho_1=\frac{1}{579}
\left[\begin{array}{c c c|c c c|c c c}
125 & \cdot & \cdot & \cdot & -5(5-12i) & \cdot & \cdot & \cdot & -5(5-12i) \\
\cdot & 125 & \cdot & \cdot & \cdot & \cdot & \cdot & \cdot & \cdot \\
\cdot & \cdot & 34 & \cdot & \cdot & \cdot & \cdot & \cdot & \cdot \\
\hline
\cdot & \cdot & \cdot & 34 & \cdot & \cdot & \cdot & \cdot & \cdot \\
-5(5+12i) & \cdot & \cdot & \cdot & 125 & \cdot & \cdot & \cdot & -5(5-12i) \\
\cdot & \cdot & \cdot & \cdot & \cdot & 125 & \cdot & \cdot & \cdot \\
\hline
\cdot & \cdot & \cdot & \cdot & \cdot & \cdot & 125 & \cdot & \cdot \\
\cdot & \cdot & \cdot & \cdot & \cdot & \cdot & \cdot & 34 & \cdot \\
-5(5+12i) & \cdot & \cdot & \cdot & -5(5+12i) & \cdot & \cdot & \cdot & 125
\end{array}\right].
\end{equation}
\end{Example}

Therefore, we show that $x$ does not have to be optimal in order to produce indecomposable witnesses, which was inconclusive in ref. \cite{P_maps}. Also, note that the state $\rho_1$ is not detected if the witness in Example \ref{Ex1} is analogously constructed from the MUBs instead.

\begin{Example}
Now, we take the $(2,N)$-POVM constructed from the orthonormal Hermitian basis presented in Appendix B, for which $x=3(5-2\sqrt{3})/4\simeq 1.15<x_{\rm opt}=3/2$ is not optimal. Assume that all the rotation matrices are $\mathcal{O}^{(\alpha)}=\mathbb{I}_2$. An indecomposable witness follows for a non-maximal number of $N=7$ POVMs and $L=3$. Indeed, one has
\begin{equation}
\widetilde{W}_2=\frac 16 \left[\begin{array}{c c c|c c c|c c c}
2+\sqrt{3} & \cdot & \cdot & \cdot & 2 & \cdot & \cdot & \cdot & 2 \\
\cdot & 5 & \cdot & \cdot & \cdot & -4 & -4 & \cdot & \cdot \\
\cdot & \cdot & 5-\sqrt{3} & -4 & \cdot & \cdot & \cdot & -4 & \cdot \\
\hline
\cdot & \cdot & -4 & 5 & \cdot & \cdot & \cdot & -4 & \cdot \\
2 & \cdot & \cdot & \cdot & 2-\sqrt{3} & \cdot & \cdot & \cdot & 2 \\
\cdot & -4 & \cdot & \cdot & \cdot & 5+\sqrt{3} & -4 & \cdot & \cdot \\
\hline
\cdot & -4 & \cdot & \cdot & \cdot & -4 & 5-\sqrt{3} & \cdot & \cdot \\
\cdot & \cdot & -4 & -4 & \cdot & \cdot & \cdot & 5+\sqrt{3} & \cdot \\
2 & \cdot & \cdot & \cdot & 2 & \cdot & \cdot & \cdot & 2
\end{array}\right],
\end{equation}
which detects a PPT state
\begin{equation}
\rho_2=\frac{1}{21}
\left[\begin{array}{c c c|c c c|c c c}
3 & \cdot & \cdot & \cdot & 1 & \cdot & \cdot & \cdot & 1 \\
\cdot & 2 & \cdot & \cdot & \cdot & 2 & 2 & \cdot & \cdot \\
\cdot & \cdot & 2 & 2 & \cdot & \cdot & \cdot & 2 & \cdot \\
\hline
\cdot & \cdot & 2 & 2 & \cdot & \cdot & \cdot & 2 & \cdot \\
1 & \cdot & \cdot & \cdot & 3 & \cdot & \cdot & \cdot & 1 \\
\cdot & 2 & \cdot & \cdot & \cdot & 2 & 2 & \cdot & \cdot \\
\hline
\cdot & 2 & \cdot & \cdot & \cdot & 2 & 2 & \cdot & \cdot \\
\cdot & \cdot & 2 & 2 & \cdot & \cdot & \cdot & 2 & \cdot \\
1 & \cdot & \cdot & \cdot & 1 & \cdot & \cdot & \cdot & 3
\end{array}\right].
\end{equation}
Hence, the full set of $N=9$ POVMs is not needed to construct indecomposable witnesses.
\end{Example}

\begin{Example}
Let us analyze an entanglement witness defined using the $(5,N)$-POVM with the associated Hermitian operator basis of the Gell-Mann matrices (see Appendix A). Also for $M=5$, the Gell-Mann basis produces symmetric measurements with the parameter $x\simeq 0.183<x_{\rm opt}=9/25$ that is not optimal. In this example, we assume that only subtractions are present ($L=N$) and that the $(M,N)$-POVM is not informationally complete ($N=1$). Now, if we take the permutation matrix
\begin{equation}
\mathcal{O}^{(1)}=\begin{pmatrix}
0 & 0 & 0 & 0 & 1 \\
1 & 0 & 0 & 0 & 0 \\
0 & 1 & 0 & 0 & 0 \\
0 & 0 & 1 & 0 & 0 \\
0 & 0 & 0 & 1 & 0
\end{pmatrix},
\end{equation}
the decomposability property of the resulting witness $\widetilde{W}_3$ is inconclusive. However, by subtracting $5(1+\sqrt{5})^2$ times more of the term for $\alpha=1$, one obtains
\begin{equation}
\widetilde{W}^\prime_3=\frac 16 \left[\begin{array}{c c c|c c c|c c c}
4 & \cdot & \cdot & \cdot & B^\ast & C & \cdot & D^\ast & B^\ast \\
\cdot & 4 & \cdot & A^\ast & \cdot & \cdot & -30i & \cdot & \cdot \\
\cdot & \cdot & 4 & 30i & \cdot & \cdot & -A^\ast & \cdot & \cdot \\
\hline
\cdot & A & -30i & 4 & \cdot & \cdot & \cdot & \cdot & \cdot \\
B & \cdot & \cdot & \cdot & 4 & \cdot & \cdot & \cdot & \cdot \\
C & \cdot & \cdot & \cdot & \cdot & 4 & \cdot & \cdot & \cdot \\
\hline
\cdot & 30i & -A^\ast & \cdot & \cdot & \cdot & 4 & \cdot & \cdot \\
D & \cdot & \cdot & \cdot & \cdot & \cdot & \cdot & 4 & \cdot \\
B & \cdot & \cdot & \cdot & \cdot & \cdot & \cdot & \cdot & 4
\end{array}\right],
\end{equation}
where
\begin{align}
A&=15(1-i)(2-i+\sqrt{5}),\\
B&=15(1-i)(2+i+\sqrt{5}),\\
C&=-30\sqrt{5}(2+\sqrt{5}),\\
D&=30(1-2i)(2+\sqrt{5}).
\end{align}
This entanglement witness is indecomposable, and it detects a PPT state
\begin{equation}
\rho_3=\frac{1}{90}
\left[\begin{array}{c c c|c c c|c c c}
10 & \cdot & \cdot & \cdot & \cdot & \cdot & \cdot & \cdot & \cdot \\
\cdot & 10 & \cdot & 3-6i & \cdot & \cdot & \cdot & \cdot & \cdot \\
\cdot & \cdot & 10 & \cdot & \cdot & \cdot & -3-6i & \cdot & \cdot \\
\hline
\cdot & 3+6i & \cdot & 10 & \cdot & \cdot & \cdot & \cdot & \cdot \\
\cdot & \cdot & \cdot & \cdot & 10 & \cdot & \cdot & \cdot & \cdot \\
\cdot & \cdot & \cdot & \cdot & \cdot & 10 & \cdot & \cdot & \cdot \\
\hline
\cdot & \cdot & -3+6i & \cdot & \cdot & \cdot & 10 & \cdot & \cdot \\
\cdot & \cdot & \cdot & \cdot & \cdot & \cdot & \cdot & 10 & \cdot \\
\cdot & \cdot & \cdot & \cdot & \cdot & \cdot & \cdot & \cdot & 10
\end{array}\right]
\end{equation}
that was undetectable by the previous (not optimal) witness $\widetilde{W}_3$.
\end{Example}

Recall that the positivity conditions for $\Phi$ in Proposition \ref{Prop} are only sufficient. Therefore, it is not surprising that we found an entanglement witness $\widetilde{W}^\prime$ which does not fall into the same class as $\widetilde{W}$ defined by such $\Phi$.

\section{Conclusions}

In this paper, we used a wide class of symmetric measurements to construct a family of positive, trace-preserving maps as well as the corresponding entanglement witnesses. Interestingly, the positivity property that is very important for the measurement operators does not determine whether the resulting map gives rise to a proper witness (has at least one negative eigenvalue). Next, it is shown that our construction belongs to the family of witnesses that are based on the CCNR separability criterion \cite{YuLiu}. The entalglement witnesses from symmetric measurements are recovered for a block-diagonal orthogonal matrix and the Hermitian orthonormal basis consisting in traceless operators and the identity. Several examples are provided using different sets of $(N,M)$-POVMs constructed from different operator bases. There remains an open question of generalizing our results to bipartite systems of different dimensions or even to a multipartite scenario.

\section{Acknowledgements}

This paper was supported by the Foundation for Polish Science (FNP) and the Polish National Science Centre project No. 2018/30/A/ST2/00837.

\bibliography{C:/Users/cynda/OneDrive/Fizyka/bibliography}
\bibliographystyle{C:/Users/cynda/OneDrive/Fizyka/beztytulow2}

\appendix

\section{Gell-Mann matrices}

In $d=3$, a popular choice of a Hermitian orthonormal basis is the Gell-Mann matrices:
\begin{align*}
g_{01}=\frac{1}{\sqrt{2}}
\begin{pmatrix}
0 & 1 & 0 \\
1 & 0 & 0 \\
0 & 0 & 0
\end{pmatrix},\qquad
&g_{10}=\frac{1}{\sqrt{2}}
\begin{pmatrix}
0 & -i & 0 \\
i & 0 & 0 \\
0 & 0 & 0
\end{pmatrix},\\
g_{02}=\frac{1}{\sqrt{2}}
\begin{pmatrix}
0 & 0 & 1 \\
0 & 0 & 0 \\
1 & 0 & 0
\end{pmatrix},\qquad
&g_{20}=\frac{1}{\sqrt{2}}
\begin{pmatrix}
0 & 0 & -i \\
0 & 0 & 0 \\
i & 0 & 0
\end{pmatrix},\\
g_{12}=\frac{1}{\sqrt{2}}
\begin{pmatrix}
0 & 0 & 0 \\
0 & 0 & 1 \\
0 & 1 & 0
\end{pmatrix},\qquad
&g_{21}=\frac{1}{\sqrt{2}}
\begin{pmatrix}
0 & 0 & 0 \\
0 & 0 & -i \\
0 & i & 0
\end{pmatrix},\\
g_{11}=\frac{1}{\sqrt{2}}
\begin{pmatrix}
1 & 0 & 0 \\
0 & -1 & 0 \\
0 & 0 & 0
\end{pmatrix},\qquad
&g_{22}=\frac{1}{\sqrt{6}}
\begin{pmatrix}
1 & 0 & 0 \\
0 & 1 & 0 \\
0 & 0 & -2
\end{pmatrix},
\end{align*}
and $G_{0}=\mathbb{I}/\sqrt{3}$. For the entanglement witness in Example 3, we fix the indices of $G_{\alpha,k}$ as follows,
\begin{equation}
\begin{split}
&G_{1,1}=g_{01},\quad G_{1,2}=g_{10},\qquad G_{2,1}=g_{02},\quad G_{2,2}=g_{20},\\
&G_{3,1}=g_{12},\quad G_{3,2}=g_{21},\qquad G_{4,1}=g_{11},\quad G_{4,2}=g_{22}.
\end{split}
\end{equation}
In Example 5, on the other hand, we take
\begin{equation}
\begin{split}
&G_{1,1}=g_{01},\quad G_{1,2}=g_{02},\quad G_{1,3}=g_{10},\quad G_{1,4}=g_{20},\\
&G_{2,1}=g_{12},\quad G_{2,2}=g_{21},\quad G_{2,3}=g_{11},\quad G_{2,4}=g_{22}.
\end{split}
\end{equation}

\section{Hermitian orthonormal basis from MUBs}

Using the complete set of four mutually unbiased bases in $d=3$ and the corresponding projectors
\begin{equation}
\begin{split}
E_{1,1}=&
\begin{pmatrix}
1 & 0 & 0 \\
0 & 0 & 0 \\
0 & 0 & 0
\end{pmatrix},\\
E_{1,2}=&
\begin{pmatrix}
0 & 0 & 0 \\
0 & 1 & 0 \\
0 & 0 & 0
\end{pmatrix},\\
E_{1,3}=&
\begin{pmatrix}
0 & 0 & 0 \\
0 & 0 & 0 \\
0 & 0 & 1
\end{pmatrix},
\end{split}\qquad
\begin{split}
E_{2,1}=&\frac 13
\begin{pmatrix}
1 & 1 & 1 \\
1 & 1 & 1 \\
1 & 1 & 1
\end{pmatrix},\\
E_{2,2}=&\frac 13
\begin{pmatrix}
1 & \omega^2 & \omega \\
\omega & 1 & \omega^2 \\
\omega^2 & \omega & 1
\end{pmatrix},\\
E_{2,3}=&\frac 13
\begin{pmatrix}
1 & \omega & \omega^2 \\
\omega^2 & 1 & \omega \\
\omega & \omega^2 & 1
\end{pmatrix},
\end{split}\qquad
\begin{split}
E_{3,1}=&\frac 13
\begin{pmatrix}
1 & \omega^2 & \omega^2 \\
\omega & 1 & 1 \\
\omega & 1 & 1
\end{pmatrix},\\
E_{3,2}=&\frac 13
\begin{pmatrix}
1 & \omega & 1 \\
\omega^2 & 1 & \omega^2 \\
1 & \omega & 1
\end{pmatrix},\\
E_{3,3}=&\frac 13
\begin{pmatrix}
1 & 1 & \omega \\
1 & 1 & \omega \\
\omega^2 & \omega^2 & 1
\end{pmatrix},
\end{split}\qquad
\begin{split}
E_{4,1}=&\frac 13
\begin{pmatrix}
1 & \omega & \omega \\
\omega^2 & 1 & 1 \\
\omega^2 & 1 & 1
\end{pmatrix},\\
E_{4,2}=&\frac 13
\begin{pmatrix}
1 & \omega^2 & 1 \\
\omega & 1 & \omega \\
1 & \omega^2 & 1
\end{pmatrix},\\
E_{4,3}=&\frac 13
\begin{pmatrix}
1 & 1 & \omega^2 \\
1 & 1 & \omega^2 \\
\omega & \omega & 1
\end{pmatrix},
\end{split}
\end{equation}
where $\omega=\exp(2\pi i/3)$, one finds the corresponding Hermitian orthonormal basis:
\begin{align*}
G_{1,1}=\frac{1}{\sqrt{3}(1+\sqrt{3})}
\begin{pmatrix}
-2-\sqrt{3} & 0 & 0 \\
0 & 1 & 0 \\
0 & 0 & 1+\sqrt{3}
\end{pmatrix},\qquad
&G_{1,2}=\frac{1}{\sqrt{3}(1+\sqrt{3})}
\begin{pmatrix}
1 & 0 & 0 \\
0 & -2-\sqrt{3} & 0 \\
0 & 0 & 1+\sqrt{3}
\end{pmatrix},\\
G_{2,1}=\frac{1}{2\sqrt{3}(1+\sqrt{3})}
\begin{pmatrix}
0 & -v^\ast & -v \\
-v & 0 & -v^\ast \\
-v^\ast & -v & 0
\end{pmatrix},\qquad
&G_{2,2}=\frac{1}{\sqrt{3}(1+\sqrt{3})}
\begin{pmatrix}
0 & iv^\ast & -iv \\
-iv & 0 & iv^\ast \\
iv^\ast & -iv & 0
\end{pmatrix},\\
G_{3,1}=\frac{1}{2\sqrt{3}(1+\sqrt{3})}
\begin{pmatrix}
0 & u^\ast & iv^\ast \\
u & 0 & -v^\ast \\
-iv & -v & 0
\end{pmatrix},\qquad
&G_{3,2}=\frac{1}{\sqrt{3}(1+\sqrt{3})}
\begin{pmatrix}
0 & u & -v^\ast \\
u^\ast & 0 & iv^\ast \\
-v & -iv & 0
\end{pmatrix},\\
G_{4,1}=\frac{1}{2\sqrt{3}(1+\sqrt{3})}
\begin{pmatrix}
0 & u & -iv \\
u^\ast & 0 & -v \\
iv^\ast & -v^\ast & 0
\end{pmatrix},\qquad
&G_{4,2}=\frac{1}{\sqrt{3}(1+\sqrt{3})}
\begin{pmatrix}
0 & u^\ast & -v \\
u & 0 & -iv \\
-v^\ast & iv^\ast & 0
\end{pmatrix},
\end{align*}
and $G_{0}=\mathbb{I}/\sqrt{3}$, where $u=(1-i)(1+\sqrt{3})$ and $v=2+\sqrt{3}+i$. The entanglement witness in Example 4 is given by eq. (\ref{WEx2}) with $G_\mu$ grouped in the following way,
\begin{equation}
\{G_1,G_2,G_3\}=\{G_{1,2},G_{2,1},G_{2,2}\},\qquad
\{G_4,G_5,G_6,G_7,G_8\}=\{G_{1,1},G_{3,1},G_{3,2},G_{4,1},G_{4,2}\}.
\end{equation}

\end{document}